\documentclass[
 reprint,
superscriptaddress,
frontmatterverbose, 
 amsmath,amssymb,
 aps,
pra,
]{revtex4-2}

\usepackage{amsthm}
\usepackage{tikz}
\usetikzlibrary{quantikz}
\usepackage{amsmath}
\usepackage{graphicx}
\usepackage{caption}
\usepackage{subcaption}
\usepackage{dcolumn}
\usepackage{bm}
\usepackage{hyperref}

\newtheorem{theorem}{Theorem}
\newtheorem{lem}{Lemma}

\begin{document}

\title{Linear-depth quantum circuits for multiqubit controlled gates}
\author{Adenilton J. da Silva}
\email{ajsilva@cin.ufpe.br}
\affiliation{Centro de Inform\'atica, Universidade Federal de Pernambuco, Recife, Pernambuco, 50740-560, Brazil}
\author{Daniel K. Park}
\email{dkd.park@yonsei.ac.kr}
\affiliation{Department of Applied Statistics, Yonsei University, Seoul, Republic of Korea}
\affiliation{Department of Statistics and Data Science, Yonsei University, Seoul, Republic of Korea}

\begin{abstract}
Quantum circuit depth minimization is critical for practical applications of circuit-based quantum computation. In this work, we present a systematic procedure to decompose multiqubit controlled unitary gates, which is essential in many quantum algorithms, to controlled-NOT and single-qubit gates with which the quantum circuit depth only increases linearly with the number of control qubits. Our algorithm does not require any ancillary qubits and achieves a quadratic reduction of the circuit depth against known methods. We show the advantage of our algorithm with proof-of-principle experiments on the IBM quantum cloud platform.
\end{abstract}

\maketitle

\section{Introduction}

Quantum computing offers exciting opportunities for a wide range of computational tasks~\cite{10.2307/2899535,zalka1998simulating,shor1999polynomial,PhysRevLett.103.150502_HHL_qBLAS,PhysRevLett.113.130503_qSQVM,Preskill2018quantumcomputing,peruzzo_variational_2014,McClean_2016,farhi2014quantum,Havlicek2019,ORUS2019100028,hur2021quantum}. 
In the circuit-based quantum computation~\cite{yao1993quantum}, efficient decomposition of quantum gates to an elementary set of one- and two-qubit gates, such as single-qubit rotation gates and the controlled-NOT (CNOT) gate, is a critical step~\cite{barenco1995elementary,saeedi2013linear} for reducing the runtime and the effect of noise~\cite{9475957,arabzadeh_depth-optimized_2013}. In particular, since the size of quantum circuits that can be implemented reliably with the near-term quantum hardware is limited due to noise~\cite{Preskill2018quantumcomputing}, minimizing the quantum circuit depth is crucial for practical applications of quantum algorithms~\cite{bae_quantum_2020}.

For many applications of quantum computing, such as quantum simulation, finance, and machine learning~\cite{PhysRevLett.113.130503_qSQVM,PhysRevX.8.041015,htc,PARK2020126422,Li_2020,Stamatopoulos2020optionpricingusing,blank_quantum-enhanced_2021,blank2022compact, schuld2018supervised}, the ability to implement multiqubit controlled unitary gates efficiently is essential. 
This is especially true for those that require encoding of classical information as probability amplitudes of a quantum state~\cite{1367-2630-17-12-123010,ffqram,9259210}. 
A seminal work presented in Ref.~\cite{barenco1995elementary} provides a systematic procedure for decomposing an $(n+1)$-qubit controlled single-qubit unitary gate, denoted by $C^nU$, to a quantum circuit whose depth and the number of one- and two-qubit gates grow quadratically with $n$. Similar results were reported in Refs.~\cite{iten2016quantum,liu2008analytic}.
In Refs.~\cite{saeedi2013linear,luo2016comment} the authors showed that the circuit depth can be reduced to grow linearly in $n$ without ancilla qubits for the special case of implementing an $n$-qubit controlled $ZX$ gate (i.e. $C^{n}ZX)$, while the number of gates remains to be quadratic in $n$. A Quantum Karnough map was introduced in Ref.~\cite{bae_quantum_2020} to reduce the number of elementary gates by a constant factor for $C^{4}X$ and $C^{5}X$, but it is not clear how to generalize the result for higher-order Toffoli gates or for an arbitrary $C^{n}U$ gate. Several works have been developed to reduce the quantum circuit depth by introducing ancillary space. For example, Ref.~\cite{lanyon_2009} showed that $C^{n}X$ can be implemented with $2n-1$ standard two-qubit gates by harnessing an $(n+1)$-level quantum system that acts as an information carrier during computation. Similarly, Ref.~\cite{he_decompositions_2017} presents the decomposition of $C^{n}X$ using $O(n)$ circuit depth and $O(n)$ elementary gates with an ancilla qubit. However, to our best knowledge, a general scheme for implementing multiqubit controlled single-qubit unitary operation using a circuit with $O(n)$ depth and $O(n^2)$ elementary gates without any ancilla qubits has not been reported. Moreover, Ref.~\cite{luo2016comment} claimed that the linear depth decomposition of $C^{n}ZX$ cannot be generalized to a decomposition of $C^nX$. In this work, we provide such a generalization, which was thought to be impossible. We generalize the result from Ref.~\cite{saeedi2013linear} and present a systematic procedure to construct linear-depth quantum circuits for $C^{n}U$, where $U$ is a $2\times 2$ unitary gate, without using any ancillary space. We also provide an implementation of our algorithm that reduces the circuit depth significantly compared to that given by the gate decomposition package in \texttt{qiskit} version 0.19.2~\cite{Qiskit} and that of Ref.~\cite{PhysRevA.71.052330}. Proof-of-principle experiments implemented on IBM quantum computers support that the circuit depth reduction achieved in this work can improve the reliability of noisy quantum devices for performing quantum algorithms.

\section{Linear-depth decomposition}

We denote $a_jUa_k$ a controlled $U$ gate with control $a_j$ and target $a_k$ and $C^{n}U$ a controlled $U$ gate with controls $a_1, \cdots, a_{n}$, and target $a_{n+1}$. With the $P$ and $Q$ gates defined in Eq.~(\ref{eq:pq}), \citet{saeedi2013linear} showed that the $(n+1)$-qubit $C^{n}R_x(\pi)$ gate  can be decomposed as in Eq.~(\ref{eq:crx})~\cite{luo2016comment}.

\begin{equation}
\begin{split}
P_n = \prod_{k=2}^{n} a_kR_x\left(\frac{\pi}{2^{n-k+1}}\right)a_{n+1} \\
Q_n = \prod_{k=1}^{n-1} C^{k}R_x(\pi)
\end{split}    
\label{eq:pq}
\end{equation}

\begin{equation}
C^{n}R_x(\pi) = Q_n^{\dagger}P_n^\dagger Q_n (a_1 R_x(\pi/2^{n-1}) a_{n+1}) P_n
\label{eq:crx}
\end{equation}

Our modification of the method proposed in Ref.~\cite{saeedi2013linear} occurs in operations applied in the target qubit. We replace the $R_x$ gates in $P_n$  by $k$th roots of operators. We apply a $\sqrt[k]{U}$ instead of $R_x(\pi/k)$ and $\sqrt[k]{U}^\dagger$ instead of $R_x(-\pi/k)$ to apply an $n$-qubit controlled $U$ gate. In a particular case, it is possible to decompose an $n$-qubit Toffoli gate with linear depth. With the operator $P_n(U)$,
$$P_n(U) = \prod_{k=2}^{n} a_k \sqrt[2^{n-k+1}]{U}a_{n+1},$$
the $n$-qubit controlled gate $C^{n}U$ is decomposed as
\begin{equation}
C^{n}U = Q_n^{\dagger}P_n(U)^\dagger Q_n (a_1 \sqrt[2^{n-1}]{U} a_{n+1}) P_n(U).
\label{eq:cru}
\end{equation}

We refer to this decomposition as the linear-depth decomposition (LDD) of $C^{n}U$. The above discussion is summarized in Theorem 1.
\begin{theorem} \label{thm:1} Any ($n$+1)-qubits controlled gate $C^{n}U$ with $U$ in the unitary group $U(2)$ can be decomposed as $Q_n^{\dagger}P_n(U)^\dagger Q_n (a_1 \sqrt[2^{n-1}]{U} a_{n+1}) P_n(U)$.
\end{theorem}
\begin{proof}
If qubits $a_1, \dots, a_{n}$ are equal to 1, all the operators in $P_n(U)$ and $a_1\sqrt[2^{n-1}]{U}a_{n+1}$ will be activated. 
Operator $Q_n$ will modify the values of qubits $a_2 \cdots a_n$ from $\ket{1}$ to $-i\ket{0}$ and the operators in $P_n(U)^\dagger$ will not be activated. 
$Q_n^\dagger$ inverts the action of $Q_n$. 
In this case, the gates $a_1\sqrt[2^{n-1}]{U}a_{n+1}P_n(U)$ are applied and they have the effect of applying a $U$ gate in the qubit $a_{n+1}$ and the other qubits are unchanged.

Now we will verify that if at least one qubit in $\{a_1, \dots, a_{n}\}$ is equal to zero the operator $Q_n^{\dagger}P_n(U)^\dagger Q_n (a_1 \sqrt[2^{n-1}]{U} a_{n+1}) P_n(U)$ will not change the input state. 
If $a_1$ is equal to zero the controlled operators $Q_n$, $Q_n^{\dagger}$ and $(a_1 \sqrt[2^{n-1}]{U} a_{n+1})$ will not be activated. 
The operators $P_n(U)$ and $P_n(U)^\dagger$ will cancel each other.

If $a_1=1$, let $a_j$ be the first qubit where $a_j$ is equal to zero. 
The operations in $P_n(U)$ controlled by $a_2, \dots, a_{j-1}$ will be activated. 
The operation controlled by $a_j$ will not be activated and activation of the rest of the gates in $P_n(U)$ depends of the value of $a_{j+1}, \dots, a_n$. 
After the application of $Q_n$ qubits $a_2, \dots, a_{j-1}$ will have the value $-i\ket{0}$, $a_j$ will have the value $-i\ket{1}$, and $a_{j+1}, \dots, a_n$ are not modified. 
Operations controlled by $a_1, \dots, a_{j-1}$ will not be applied in $P_n(U)^\dagger$; operations controlled by $a_j$ will be activated; and because $a_{j+1}, \dots, a_n$ are not modified the operations controlled by them in $P_n(U)^\dagger$ will cancel out with the operations of $P_n(U)$ with the same controls. 
$Q^\dagger$ inverts $Q$ and the overall action of $Q_n^{\dagger}P_n(U)^\dagger Q_n (a_1 \sqrt[2^{n-1}]{U} a_{n+1}) P_n(U)$ is to apply the following operation in the $(n+1)$th qubit.
$$\sqrt[2^{n-1}]{U} \sqrt[2^{n-1}]{U} \sqrt[2^{n-2}]{U} \sqrt[2^{n-3}]{U} \cdots \sqrt[2^{n-j+2}]{U}\sqrt[2^{n-j+1}]{U}^\dagger = I$$
\end{proof}

\begin{theorem}
$Q_n^{\dagger}P_n(U)^\dagger Q_n (a_1 \sqrt[2^{n-1}]{U} a_{n+1}) P_n(U)$ can be implemented in a circuit with linear depth.
\end{theorem}

\begin{proof}
$P_n(U)$ and $P_n(U)^\dagger$ can be implemented with controlled-NOT and single qubit gates in a circuit with linear depth.
We only need to verify if $Q_n$ can be implemented in a linear depth circuit. We can decompose $Q_n$ as
$$Q_n = Q_{n-1}C^{n-1}R_x(\pi),$$
and replacing $C^{n-1}R_x(\pi)$ with the decomposition in Theorem~\ref{thm:1},
$C^{n-1}R_x(\pi) = Q_{n-1}^{\dagger}P_{n-1}^\dagger Q_{n-1} (a_1 R_x(\frac{\pi}{2^{n-1-1}}) a_n) P_{n-1},$
$Q_{n-1}^\dagger$ will cancel the action of $Q_{n-1}$ and we obtain
$$
Q_n = P_{n-1}^\dagger Q_{n-1} \left(a_1 R_x\left(\frac{\pi}{2^{n-1-1}}\right) a_{n}\right) P_{n-1}.
$$

Recursively replacing $Q_{n-1}$ in the last decomposition of $Q_{n}$, we obtain 
\begin{equation*}
\begin{split}
Q_n= P_{n-1}^\dagger P_{n-2}^\dagger Q_{n-2} \left(a_1 R_x\left(\frac{\pi}{2^{n-2-1}}\right) a_{n-1}\right) P_{n-2} \\ \left(a_1 R_x\left(\frac{\pi}{2^{n-1-1}}\right) a_{n}\right) P_{n-1}.
\end{split}
\end{equation*} 
The structure of $Q_n$ after the $i$th decomposition of the operators $Q$ is described in~Eq.~(\ref{eq:qndecomposition}).
\begin{equation}
\label{eq:qndecomposition}
Q_n = \prod_{k=n-1}^{n-i} P_k^\dagger Q_{n-i} \prod_{j=n-i}^{n-1} a_1 R_x\left(\frac{\pi}{2^{j-1}}\right)a_{j+1} P_j.
\end{equation}

The last recursive call occurs with $i=n-2$. Because $Q_{n-(n-2)} = Q_2$ that contains only a $C^1R_x(\pi)$ gate, after the last recursive call, we obtain 
\begin{equation}
\label{eq:qn}
Q_n = \prod_{k=n-1}^{2} P_k^\dagger Q_2 \prod_{j=2}^{n-1} a_1 R_x\left(\frac{\pi}{2^{j-1}}\right)a_{j+1} P_j.
\end{equation}

Because $a_1 R_x(\frac{\pi}{2^{j-1}})a_{j+1} P_j$ is a $P_{j+1}$ gate, and $P_k^\dagger$ has the same depth of a $P_k$ gate, we only need to show that $\prod_{k=2}^{n-1} P_k$ has linear depth. $\prod_{k=2}^{n-1}P_k = \prod_{k=2}^{n-1} P_k(R_x(\pi))$ and Lemma~\ref{lem:depth} shows that $\prod_{k=2}^{n-1} P_k(U_k)$ has a linear depth.
\end{proof}

\begin{lem}
\label{lem:depth}
There is a circuit with single-qubit gates and controlled two-qubit gates that implements $\prod_{k=2}^n P_k(U_k)$ with a circuit depth equal to $2n-3$ with $n\geq 3$, where $U_k$ are single-qubit gates.
\end{lem}

\begin{proof}
We apply induction on $n$. For $n=3$, $P_2(U_2)$ has depth 1, $P_3(U_3)$ has depth 2, and $P_2(U_2)P_3(U_3)$ has depth 3 that is equal to $2\cdot 3 - 3$.

The operation $\prod_{k=2}^{n+1}P_k(U_k)$ can be decomposed as $\prod_{k=2}^{n}P_k(U_k) P_{n+1}(U_{n+1})$. With $3 \leq j \leq n$, the $j$th gate of $P_{k+1}(U_{k+1})$ operates on qubits $k+2$ and $k-j+2$ and can be applied in parallel with the $(j-2)$th gate of $P_k(U_k)$ that operates on qubits $k+1$ and $k-j+3$. 
In this way, only two gates of $P_{k+1}(U_{k+1})$ cannot be applied in parallel with other gates of  $\prod_{k=2}^{n}P_k(U_k)$ and $depth(\prod_{k=2}^{n+1}P_{k}(U_k))= depth(\prod_{k=2}^{n}P_k(U_k))+2$. 
By the induction hypothesis $depth(\prod_{k=2}^{n+1}P_k(U_k)) = 2n-3+2 = 2(n+1)-3$.
\end{proof}

We can implement $P_n(U)^\dagger Q_n (a_1 \sqrt[2^{n-1}]{U} a_{n+1}) P_n(U)$ in a circuit with depth $(2(n+1) - 3) + (2n-3)$. With the decomposition of $Q_n$ presented in Eq.~\ref{eq:qn} the gates $P_n(U)^\dagger \prod_{k=n-1}^{2}P_k^\dagger$ have depth $2n-3$,  and  the remaining gates 
$$Q_2 \prod_{j=2}^{n-1} a_1 R_x\left(\frac{\pi}{2^{j-1}}\right)a_{j+1} P_j P_n(U)$$ 
are equivalent to $\prod_{j=2}^{n+1}P_k(U_k)$ with $U_2= R_x(2\pi)$, $U_{n+1}=U$,  $U_k = R_x(\pi)$ for $k\not\in \{2, n+1\}$ and depth $(2(n+1) - 3)$.  

$Q_n^{\dagger}$ is equivalent to $(\prod_{k=n-1}^{2}P_k\prod_{k=2}^{n}P_k(U_k))^\dagger$ with $U_2=R_x(2\pi)$ and $U_k=R_x(\pi)$ for $k\neq 2$ and can be implemented in a circuit with depth $(2(n-1)-3) + 2n-3$. The overall depth of a $C^nU$ gate is $8n-12$ or $8m-20$, where $m=n+1$ is the number of qubits in $C^nU$.
Figure~\ref{fig:cnu} shows the decomposition of a $C^5U$ gate.

We show that an $n$-qubit controlled single-qubit unitary gate can be decomposed into a circuit with a linear depth of two-qubit controlled gates. Each two-qubit controlled gate requires a fixed number of single-qubit and CNOT gates \cite[Lemma 5.1]{barenco1995elementary}; then the circuit can be decomposed with a linear depth of single-qubit and CNOT gates. 

\begin{figure*}[t]
    \centering
    
    \includegraphics[width=\textwidth]{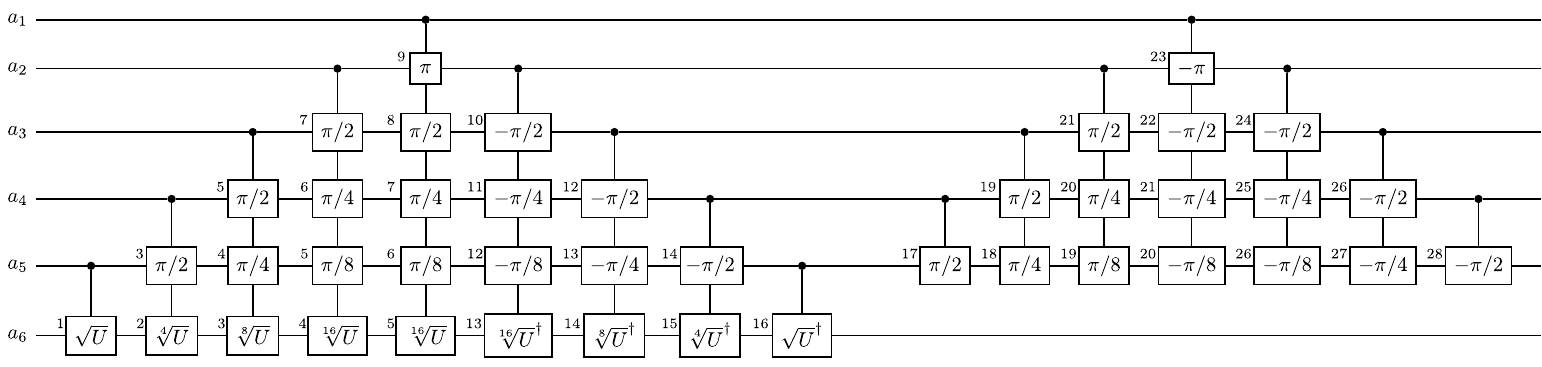}

    \caption{A linear-depth decomposition of a $C^5U$. The number $\pi/k$ in boxes indicate the angle of single-qubit rotation around the x-axis of the Bloch sphere (i.e. $R_x(\pi/k)$). The circuit modification to apply a $C^nU$ gate occurs in the last qubits where we apply a $\sqrt[k]{U}$ instead of a $R_x(\pi/k)$ and $\sqrt[k]{U}^\dagger$ instead of a $R_x(-\pi/k)$. The numbers in the left side of each gate corresponds to the time step at which the corresponding gate is applied. Note that setting $U=X$ produces a 6-qubits Toffoli gate.}
    \label{fig:cnu}
\end{figure*}

\section{Numerical analysis}
\label{sec:simulation}
We evaluate the depth of a quantum circuit produced by the LDD for a $C^{n}U$ gate with a randomly selected $2 \times 2$ unitary gate with the total number of qubits, $n+1$, ranging from 2 to 13. To demonstrate the advantage of LDD, we compare the results against those produced by the default implementation of a multiqubit controlled unitary operation in \texttt{qiskit} and by the algorithm for uniformly-controlled one-qubit gates (i.e. multiplexer)~\cite{PhysRevA.71.052330}. The latter is also implemented with \texttt{qiskit}. The results are plotted in Fig.~\ref{fig:2}. In the figure, the default \texttt{qiskit} implementation, uniformly-controlled gates, and our LDD algorithm are indicated by the diamonds, circles, and squares and are labelled as Qiskit (naive), Qiskit (uc), and Linear, respectively. The figure also shows two different cases that we tested, first without any further circuit optimization (solid lines) and second with another layer of circuit optimization provided by the transpilation tool in \texttt{qiskit} (dashed lines). 
\begin{figure}[ht]
    \centering
    \includegraphics[width = 0.9\columnwidth]{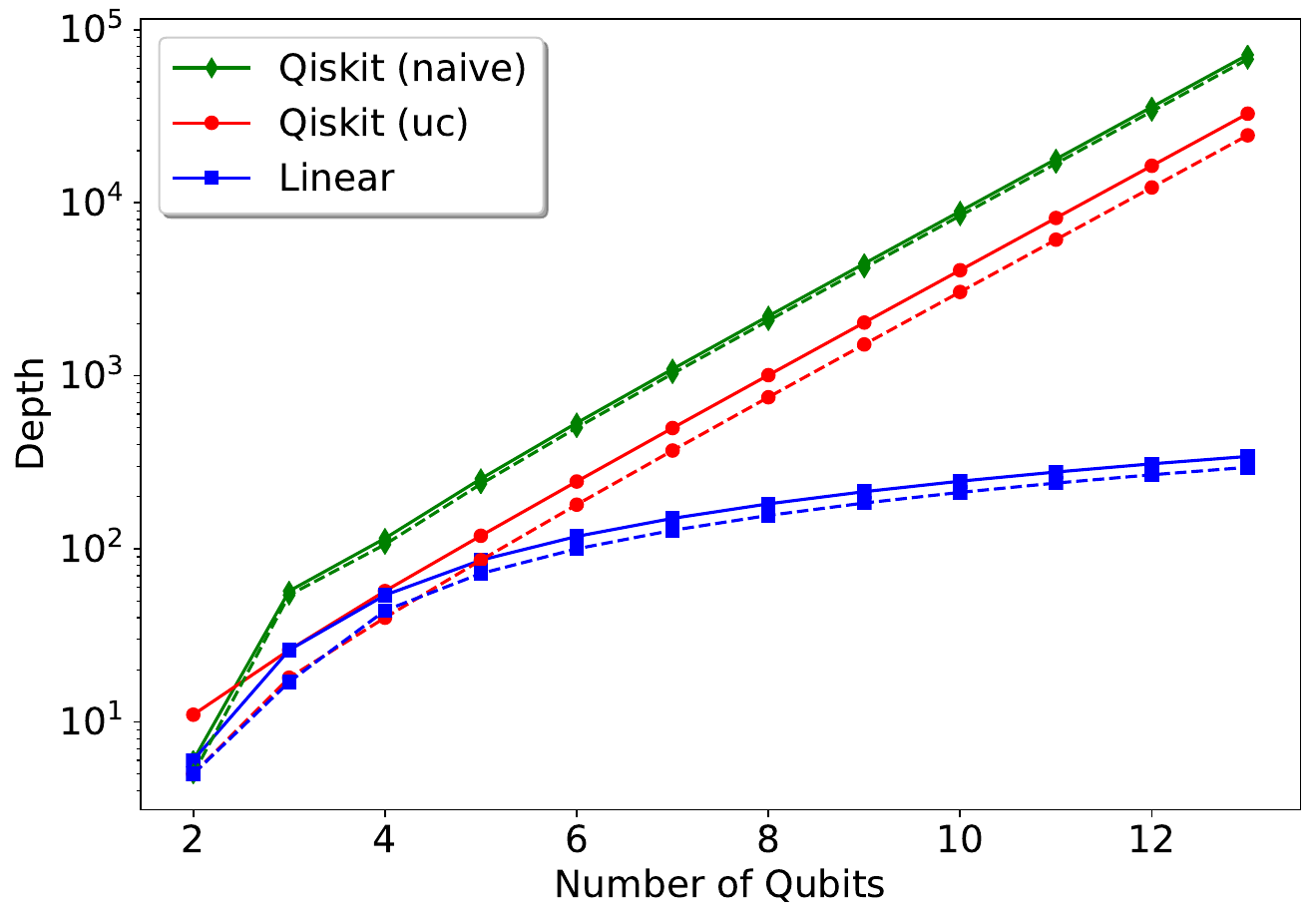}
    \caption{\label{fig:2}Numerical analysis of the depth of quantum circuits produced by three different gate decomposition methods: the default implementation of multiqubit controlled unitary operation in \texttt{qiskit} and by the algorithm for uniformly-controlled one-qubit gates. The resulting circuit depths from these cases are shown with the diamond, circle, and square symbols, respectively. The solid lines are obtained without any optimization, and the dashed lines are obtained after circuit optimization provided by Qiskit (with optimization level 3).}
\end{figure}

The numerical analysis confirms that the LDD algorithm is more economical than the previous methods in terms of the quantum circuit depth. The quantum circuit depth increases linearly with the number of qubits for LDD, and exponentially for other methods. This is especially advantageous when the number of qubits in the system is more than four.

\section{Computational experiments}
\label{sec:exp}
Minimizing the quantum circuit depth is crucial for noisy intermediate-scale quantum (NISQ) computing~\cite{Preskill2018quantumcomputing} without fault tolerance and error correction. To demonstrate the advantage of our LDD for implementing quantum algorithms on a noisy quantum hardware, we performed two proof-of-principle computational experiments on IBM cloud quantum devices. The first experiment tests $C^{n}X$. Since quantum process tomography becomes extremely costly as the number of qubits increases~\cite{NC00}, we use a simpler experiment that allows us to compare different decomposition methods. Namely, we design a quantum circuit that estimates
\begin{equation}
\label{eq:exp1}
    |\langle 11\ldots 11| C^{n}X |11\ldots 10 \rangle |^2,
\end{equation}
which should be 1 in the ideal case. The quantum circuit for testing Eq.~(\ref{eq:exp1}) is shown in Fig.~\ref{fig:circuit} (a). Similarly, the second experiment is designed to estimate
\begin{equation}
\label{eq:exp2}
    |\langle 11\ldots 11| C^{n}U (I^{\otimes n}\otimes U^{\dagger}) |11\ldots 11 \rangle |^2,
\end{equation}
where $I^{\otimes n}$ indicates the identity operation applied to $n$ qubits and $U$ is a randomly chosen $2 \times 2$ unitary gate. The above probability should be 1 in the ideal case. The quantum circuit for testing Eq.~(\ref{eq:exp2}) is shown in Fig.~\ref{fig:circuit} (b). Note that for both experiments, if a completely depolarizing channel is applied instead of the desired gate operation, then the above probabilities become $1/2^{n+1}$.

\begin{figure}[ht]
    \centering
    \includegraphics[width = 0.9\columnwidth]{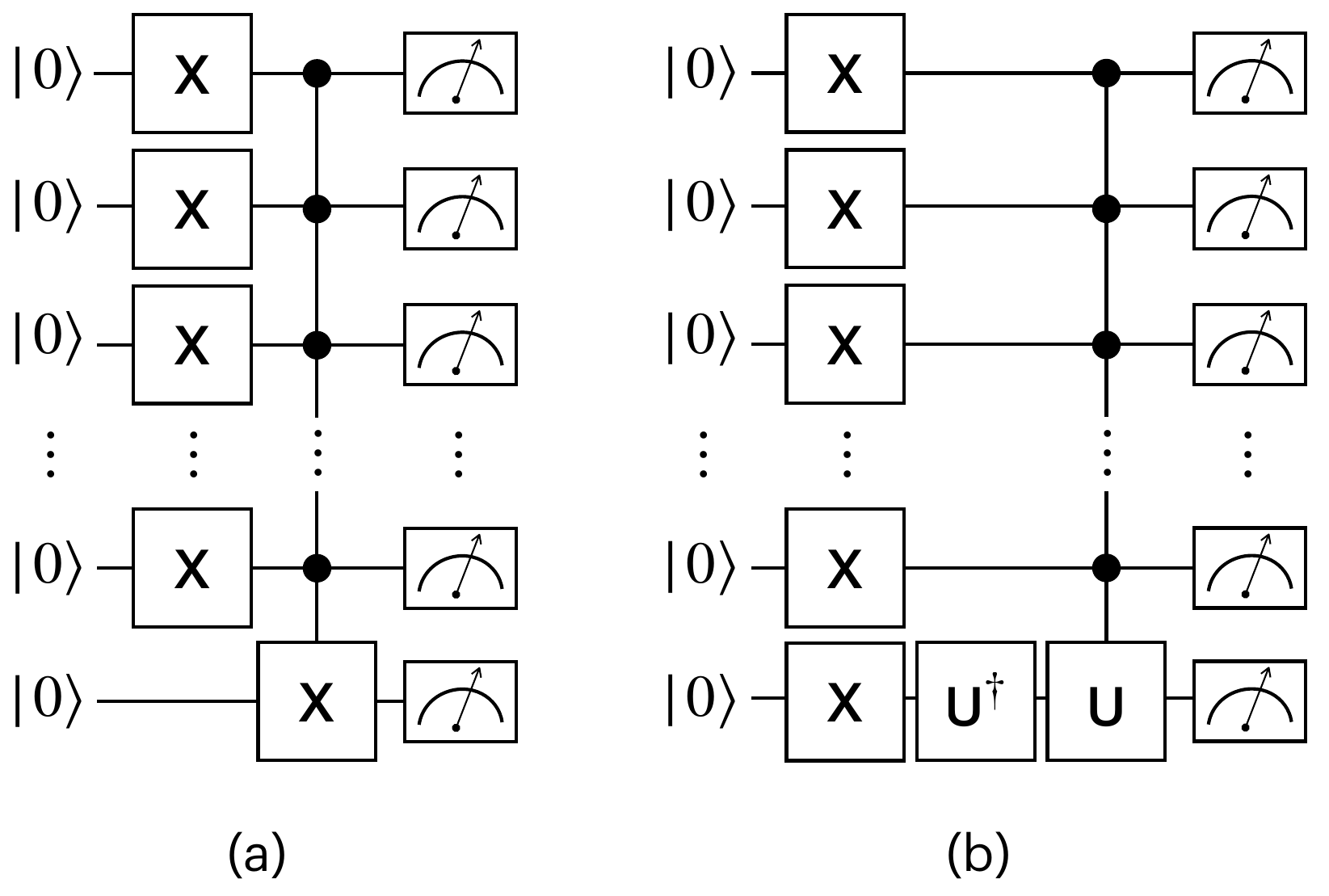}
    \caption{\label{fig:circuit} Quantum circuits tested in proof-of-principle experiments.}
\end{figure}

The two experiments described above were executed on IBM quantum devices through the IBM Quantum cloud service~\cite{IBMQE}. The first experiment was implemented on \texttt{ibm\_hanoi}, a 27-qubit Falcon r5.11 processor, and the second one was implemented on \texttt{ibmq\_guadalupe}, a 16-qubit Falcon r4p processor. 
The selection was based on the amount of queue on the cloud service at the time of execution. Quantum volumes~\cite{PhysRevA.100.032328} for these devices are 64 and 32, respectively, and the device diagrams with qubit connectivities are provided as insets in Fig.~\ref{fig:experiment}. Typical average error rates in \texttt{ibm\_hanoi} are $1.9\times 10^{-4}$ for single-qubit Pauli X gates, $4.4\times 10^{-2}$ for CNOT gates and $1.8\times 10^{-2}$ for readout. The average T$_1$ and T$_2$ relaxation times typically measured on this device are 175 $\mu$s and $140$ $\mu$s, respectively. Typical average error rates in \texttt{ibm\_guadalupe} are $3.1\times 10^{-4}$ for single-qubit Pauli X gates, $9.4\times 10^{-2}$ for CNOT gates and $1.7\times 10^{-2}$ for readout. The average T$_1$ and T$_2$ relaxation times typically measured on this device are $103$ $\mu$s and $102$ $\mu$s, respectively.

Similar to the numerical analysis presented in the previous section, the quantum circuit produced by our LDD algorithm is compared with those produced by the default implementation in \texttt{qiskit} and by the algorithm for uniformly-controlled one-qubit gates. 
To execute quantum circuits on real quantum devices, they need to be decomposed further with respect to the native gate set and the qubit connectivity of the target device. Thus the final quantum circuits are obtained after performing the \texttt{qiskit} transpilation tool with optimization level 2. The probabilities shown in Eqs.~(\ref{eq:exp1}) and~(\ref{eq:exp2}) are estimated by sampling measurement outcomes from the quantum circuits shown in Figs.~\ref{fig:circuit}(a) and (b) 50000 and 32000 times, respectively, for \texttt{ibm\_hanoi} and \texttt{ibmq\_guadalupe}.
The experimental results are shown in Fig.~\ref{fig:experiment}. In the figure, the default \texttt{qiskit} implementation, multiplexer, and our LDD algorithm are indicated by triangles, circles, and squares and labelled as naive, UCG, and LDD, respectively. Due to the finite number of sampling, the probability is zero for some instances, especially when the number of qubits is large. Hence certain values are missing. To exhibit the connection between the circuit depth and the accuracy of the quantum circuit execution, Fig.~\ref{fig:experiment} also shows the quantum circuit depth.

\begin{figure*}%
\begin{subfigure}[b]{0.49\textwidth}
         \centering
         \includegraphics[width=\textwidth]{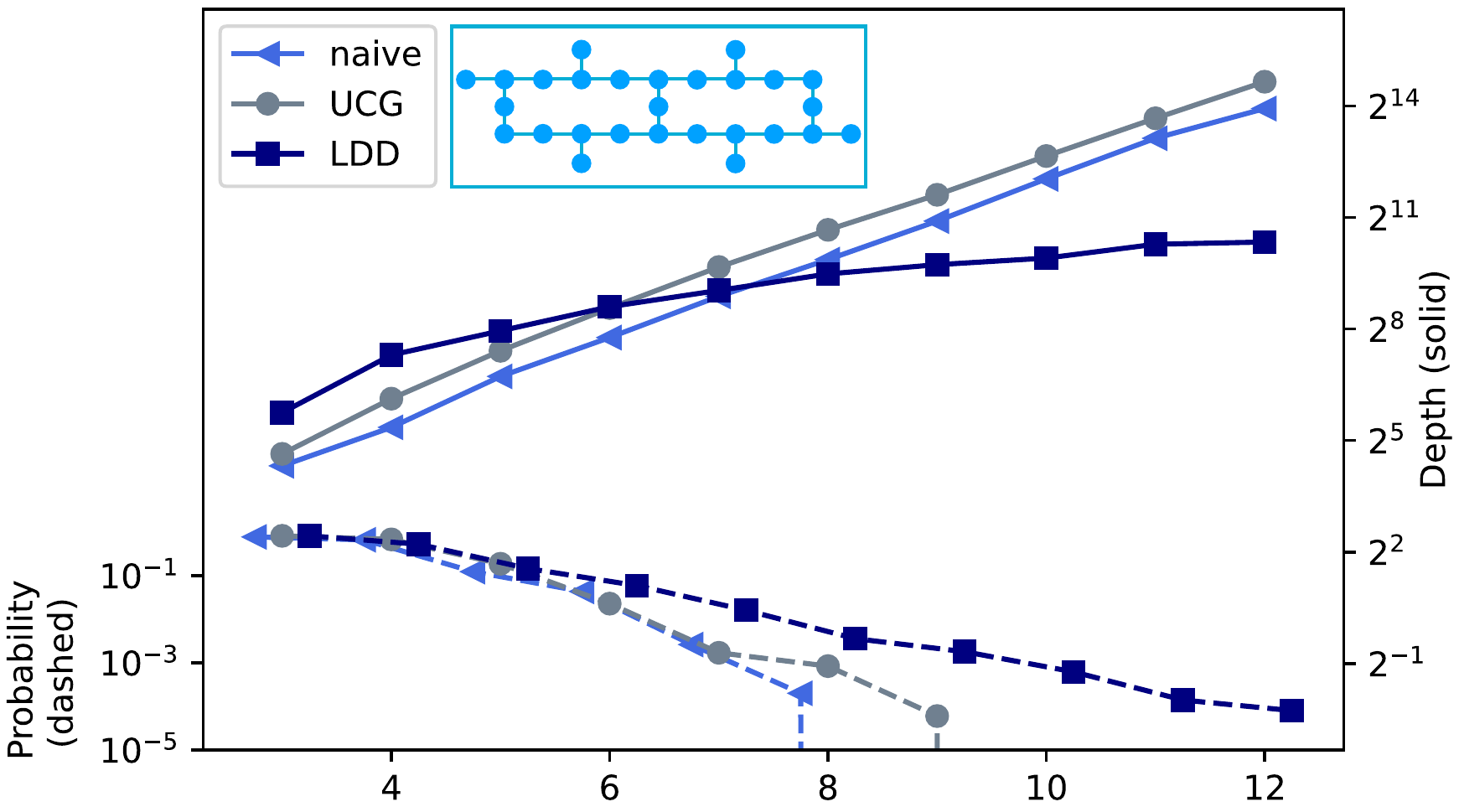}
         \caption{}
         \label{fig:a}
     \end{subfigure}
     \hfill
     \begin{subfigure}[b]{0.49\textwidth}
         \centering
         \includegraphics[width=\textwidth]{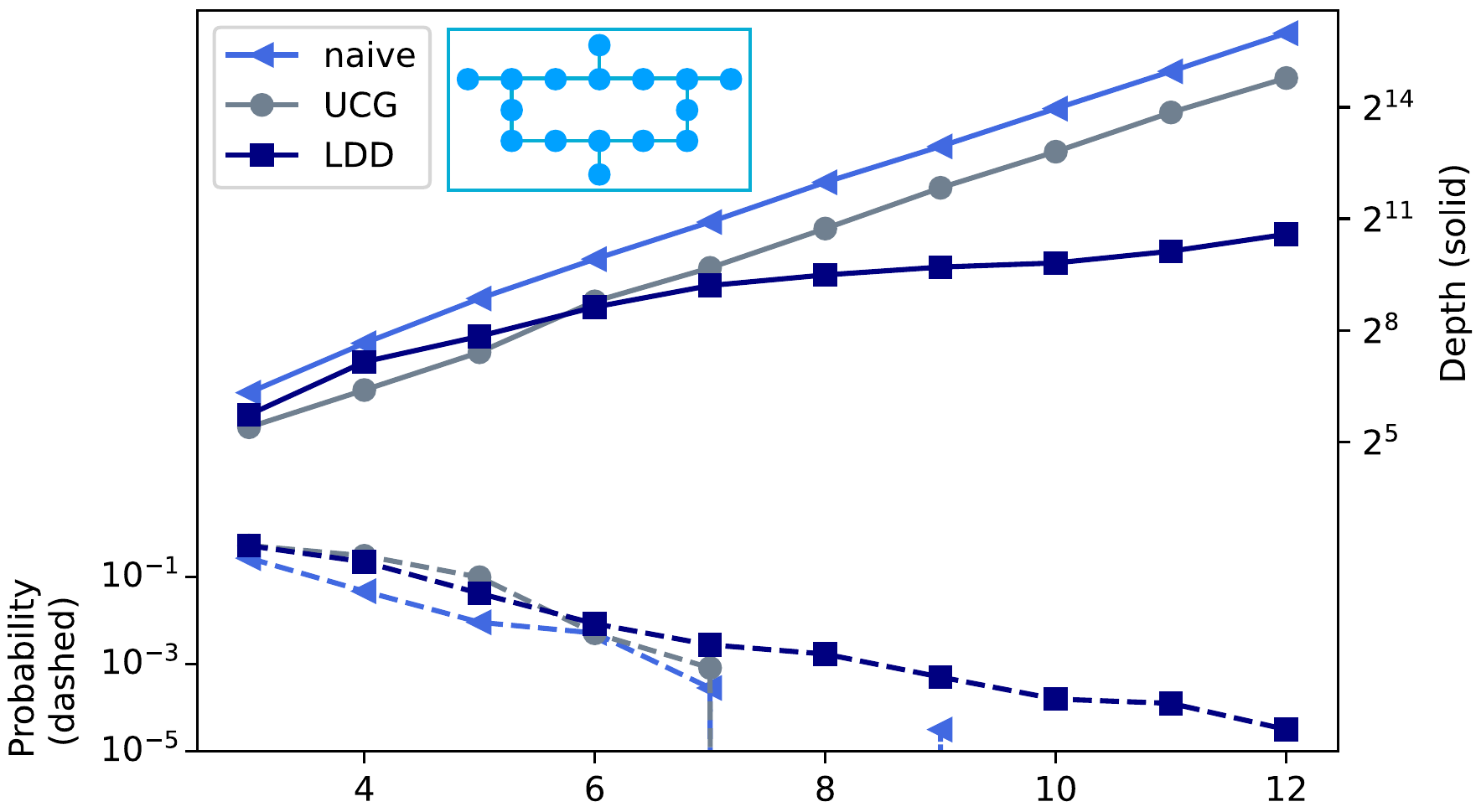}
         \caption{}
         \label{fig:b}
     \end{subfigure}
		\caption{\label{fig:experiment} Results from proof-of-principle experiments implemented on (a) \texttt{ibm\_hanoi} for $C^{n}X$ and (b) \texttt{ibmq\_guadalupe} for $C^{n}U$, where $U$ is randomly chosen from the unitary group $U(2)$. The final quantum circuits are obtained after performing the \texttt{qiskit} transpilation tool with optimization level 2. The diamond, circle, and square symbols represent the results from the default implementation of \texttt{qiskit}, uniformly-controlled one-qubit gates, and the linear-depth method. The dashed lines indicate the probabilities of measuring a string of ones, and the solid line is the circuit depth. The probabilities are obtained from 50000 and 32000 samples, respectively for (a) and (b). The solid lines vanishing towards the x-axis means that the probability is zero. For each experiment, the device diagram illustrating the qubit connectivity is shown next to the legend labels.}
\end{figure*}

As predicted by the theory and demonstrated by the numerical results shown in Fig.~\ref{fig:experiment}, the quantum circuit depth of the LDD algorithm increases much more slowly with the number of qubits than that of previous methods. For $C^{n}X$, the LDD method yields a longer quantum circuit depth than the other methods for small $n$ until the break-even point at $n=6$. Afterwards, the quantum circuit depth is shallower than the other methods. Similarly, for $C^{n}U$, the LDD method yields a longer quantum circuit depth than the multiplexer until the break-even point at $n=6$. Beyond this point, the quantum circuit depth is shallower than the other methods. Consequently, the LDD method begins to produce higher probability than those of the other methods after the break-even point. The proof-of-principle experiments confirm that the quantum circuit depth reduction achieved by our work can help achieve more accurate executions of quantum algorithms on NISQ devices.

\section{Conclusion}
In summary, we showed that an $n$-qubit controlled single-qubit unitary gate can be decomposed into a circuit with a linear depth of single-qubit and CNOT gates. The linear-depth decomposition does not require any ancilla qubits. Our method starts to outperform the circuit decomposition used in \texttt{qiskit}~\cite{cross2018ibm}, a publicly available tool, when the number of qubits is five, and the improvement increases with the number of qubits. 
Through numerical analysis and proof-of-principle experiments performed on the IBM quantum cloud platform, we verified the advantage of the proposed method.  

Controlled operators are the basic building block of quantum algorithms, such as the implementation of isometries~\cite{iten2016quantum}, quantum machine learning~\cite{htc}, quantum finance~\cite{blank_quantum-enhanced_2021}, and state preparation~\cite{plesch2011quantum}. A more efficient controlled operation should allow improvements in several quantum computing applications. A possible future work is to use this alternative decomposition to investigate how to reduce the depth of a sequence of multiqubit controlled gates, improving the application of quantum multiplexers.

\section*{Data availability}
An implementation of the proposed method is publicly available at \url{https://github.com/qclib/qclib/blob/master/qclib/gates/mc_gate.py}.

\section*{Acknowledgment}

This work is supported by Brazilian research agencies Conselho
Nacional de Desenvolvimento Científico e Tecnológico - CNPq (Grant No. 308730/2018-6), Coordenação de Aperfeiçoamento de Pessoal de Nível Superior (CAPES) - Finance Code 001 and Fundação de Amparo à Ciência e Tecnologia do Estado de Pernambuco - FACEPE (APQ-1229-1.03/21). D.K.P. acknowledges support from the National Research Foundation of Korea (Grant No. 2019R1I1A1A01050161 and No. 2022M3E4A1074591) and the KIST Institutional Program (2E31531-22-076). We acknowledge the use of IBM Quantum services for this work. The views expressed are those of the authors, and do not reflect the official policy or position of IBM or the IBM Quantum team.

\end{document}